\documentclass{article}
\usepackage[preprint]{neurips_2018}
\usepackage[utf8]{inputenc} 
\usepackage[T1]{fontenc}    
\usepackage{hyperref}       
\usepackage{url}            
\usepackage{booktabs}       
\usepackage{amsfonts}       
\usepackage{bm}
\usepackage{nicefrac}       
\usepackage{microtype}      
\usepackage{wrapfig}
\usepackage{lipsum}
\usepackage{amsmath}
\usepackage{enumitem}
\usepackage{amsthm}
\usepackage{amssymb}
\usepackage{algorithm}
\usepackage{algpseudocode}
\usepackage{xcolor}
\usepackage{subfigure}
\newtheorem{theorem}{Theorem}
\newtheorem{proposition}{Proposition}
\newtheorem*{proposition_noindex}{Proposition}

\newtheorem{definition}{Definition}
\newtheorem{assumption}{Assumption}

\newcommand{\ymax}{y_{\max}}
\usepackage{color}

\usepackage{graphicx}

\title{No-Regret Learning in Cournot Games}

\author{%
	Yuanyuan Shi, Baosen Zhang\\
	Department of Electrical and Computer Engineering\\
	University of Washington\\
	\texttt{yyshi, zhangbao@uw.edu} \\
}

\begin{document}
	
	\maketitle

\begin{abstract}
This paper examines the convergence of no-regret learning in Cournot games with continuous actions. Cournot games are the essential model for many socio-economic systems, where players compete by strategically setting their output quantity. We assume that players do not have full information of the game and thus cannot pre-compute a Nash equilibrium. Two types of feedback are considered: one is bandit feedback and the other is gradient feedback. To study the convergence of the induced sequence of play, we introduce the notion of convergence in measure, and show that the players’ actual sequence of action converges to the unique Nash equilibrium. In addition, our results naturally extend the no-regret learning algorithms’ time-average regret bounds to obtain the final-iteration convergence rates. Together, our work presents significantly sharper convergence results for learning in games
without strong assumptions on game property (e.g., monotonicity) and shows how exploiting the game information feedback can influence the convergence rates.
\end{abstract}

\section{Introduction}
\label{sec:intro}
Game-theoretic models have been used to describe the cooperative and competitive behaviors of a group of players in various systems, including robotics, distributed control and socio-economic systems~\cite{li2019differential,marden2015game,lanctot2017unified,serrino2019finding}.
In this paper, we study the interaction of dynamic learning agents in Cournot games~\cite{cournot1838recherches}.
Cournot game is the essential market model for many critical infrastructure systems such as the energy systems~\cite{cai2019role}, transportation networks~\cite{Kostas14ec} and healthcare systems~\cite{chletsos2019hospitals}. It is also one of the most prevalent models of firm competition in economics. In the Cournot game model, firms control their production level, which influences the market price~\cite{nadav2010no}.
For example, most of the US electricity market is built upon the Cournot model~\cite{kirschen2004fundamentals}, where the generators bid to serve the demand in the power grid, and the electricity price is decided by the total amount of supply and demand. Each generator's payoff is calculated as the market price multiplying its share of the supply deducting any production cost. The goal of each player is to maximize its payoff by strategically choosing the production quantity.

Previous works in Cournot games mostly focus on analyzing the equilibrium behavior, especially the Nash equilibrium. Nash equilibrium describes a stable state of the system, as no player can gain anything by changing only their own strategy. However, it is not clear how players can reach the equilibrium if they do not start from one. Specifically, in the context of Cournot games, each player only has access to local information (i.e., his own payoff), and is not informed about the attributes of other participants. Thus, they cannot pre-compute, or agree on a Nash equilibrium before they begin interacting. For these reasons, in this work, we move away from the static view of assuming players are at a Nash equilibrium. Instead, we analyze the long-run dynamics of learning agents in Cournot games, and ask the following two fundamental questions:
\emph{\begin{enumerate}
		\itemsep0.3em
		\item Will strategic learning agents reach an equilibrium?
		\item If so, how quickly do they converge to the equilibrium?
\end{enumerate}}
Reasoning about these questions requires specifying the dynamics, which describes how players adapt their actions before reaching the equilibrium. In particular, we consider the dynamics induced by no-regret learning algorithms. We focus on no-regret learning algorithms for two main reasons. Firstly, it is a sensible choice for decision makers under \emph{limited information} situations, as no one wants to realize that in hindsight the policy they employed is strictly inferior to a simple policy that takes
the same action throughout. Secondly, it allows us to make \emph{minimal assumptions} on the players' decision policy. Different players can choose different algorithms, as long as the policy satisfies the no regret guarantee. In fact, no-regret learning is an active area of research and includes a wide collection of algorithms, and interested readers can refer to~\citep{gordon2007no,hazan2016introduction} for a detailed review.

\subsection{Related Works}
Dynamic behavior in Cournot games has been studied before. Cournot~\cite{cournot1838recherches} considered the simple best response dynamics in his original paper, where players react to their opponents’ actions on the previous step. Cournot proved that the best response dynamics converge to the unique Nash equilibrium (NE), after sufficiently many steps. However,
his results only apply to two-player games. Since then, a rich set of literature has tried to generalize this result. For example, \cite{Milgrom1990} proposed an adaptive behavior rule to reach the NE for an arbitrary number of players in Cournot games.~\cite{roughgarden2016twenty} showed that in Cournot games with linear price or cost function (hence a potential game), the best response dynamics converge to the NE. However, all these best response dynamics require full information of other players' previous actions and the exact game structure (i.e., price function, payoff functions). However, for many applications of interest, providing full feedback is either impractical such as distributed control~\cite{marden2009overcoming} or explicitly disallowed due to privacy and market power concerns such as energy markets~\cite{ito2016sequential}.

Studying the behavior of players under limited information is challenging, which starts to receive some recent attention. Most literature considers the no-regret dynamics because of their inherent robustness under uncertainty~\cite{bravo2018bandit}.
A general result about no-regret learning in games is that if all players experience no-regret as the time approaches infinity, the time-average action converges to a coarse correlated equilibrium~\cite{roughgarden2016twenty, syrgkanis2015fast, foster2016learning}. However, coarse correlated equilibria (CCE) is a loose notion of equilibrium and may contain actions that are manifestly suboptimal for players~\cite{barman2015finding}. In addition, besides the time-average behavior, players also care about their actual (or so called the final-iteration) behavior.

Regarding the final-iteration convergence to a finer notion of equilibrium (e.g., Nash), assumptions on either game structure or player's action strategy, or both, are usually made. On the game structure side,~\cite{bravo2018bandit} showed that the convergence of online mirror descent (a no-regret algorithm) with bandit feedback to the NE, in strongly monotone games\footnote{We will formally define monotone games in Section 2.3.}.~\cite{cohen2017learning} proved the convergence of multiplicative weights (also a no-regret algorithm) and~\cite{perkins2017mixed} proved the convergence of actor-critic reinforcement learning algorithm, both in potential games. However, a general Cournot game does not satisfy the monotone game and potential game definition, without further restrictions on the allowed price function class (e.g., linear). On the action strategy side,
~\cite{zhou2017mirror,zhou2018learning,mertikopoulos2019learning} relaxed the game structure assumption to variational stable games, which includes a broader range of games. But their proof only applies to online mirror descent (OMD) algorithm and does not generalize to all no-regret dynamics. Similarly,~\cite{bervoets2016learning} assumed a specific action strategy and proved its convergence in games with concave payoff functions. In summary, establishing general convergence in Cournot games under minimal game structure and behavioral hypotheses remains open.

In the most relevant result to our work,~\cite{nadav2010no} studied the no-regret dynamics in Cournot games with linear price functions and convex cost function. They proved the final-iteration convergence to the NE. However, as we will discuss in Section 2.3, Cournot games with their assumptions are \emph{monotone} games thus the analysis are greatly simplified.
Our work is a strict generalization of their results since we allow for a broad class of concave price functions. In fact, to the best of our knowledge, our work is one of the few that obtain positive convergence results without the monotonicity or stronger assumptions in the game structure.

\subsection{Our Contributions} In this work, we study the dynamics of no-regret learning algorithms in Cournot games, and our major contributions are in the following three aspects.

\textbf{Firstly}, we introduce a convergence notion which we call \emph{convergence in measure}. This convergence notion extends the standard convergence notion and permits negligible behavior variations (i.e., actions in a measure zero set). In fact, this notion allows us to treat the convergence question for general no-regret algorithms, without having to restrict ourselves to a specific subclass (such as online mirror descent~\cite{zhou2017mirror,mertikopoulos2019learning}).

\textbf{Secondly}, we provide a detailed analysis of the long-run dynamics of no-regret learning in Cournot games. We prove that both the time-average and final-iteration converge to the unique Nash equilibrium. The latter result on the limiting action of players has not be shown before.  We have no strong assumptions on the game structure (e.g., monotonicity) and no restrictions for the no-regret algorithm type. This is a much sharper result compared to the existing results on the time-average behavior converging to a coarse correlated equilibrium~\cite{roughgarden2016twenty} or the final-iteration convergence with specific game structure assumptions~\cite{bravo2018bandit,cohen2017learning}.

\textbf{Thirdly}, we derive the final-iteration convergence rate under the notion of convergence in measure, and link it to the time-average regret bounds of different no-regret algorithms. For concreteness, we show that convergence rate of zeroth-order FKM algorithm~\cite{flaxman2005online} is $O(T^{-1/4})$ and that of first-order OMD algorithm~\cite{shalev2007online} is $O(T^{-1/2})$. It provides quantitative insights for market designers on the benefit of releasing more information to the players, in terms of the market equilibration rate.

\section{Problem Setup and Preliminaries}
\label{sec:model}
In this section, we first introduce the Cournot game model and assumptions. Then, we provide two motivating examples of its applications in social infrastructure. Next, we review different types of no-regret learning algorithms and discuss the existing convergence results and the difficulties of convergence analysis in Cournot games.
\subsection{Model of Cournot Game}
\begin{definition}[Cournot game~\cite{cournot1838recherches}]
	Consider $N$ players produce homogenous product in a limited market, where the action space of each player is the production level $\forall i, x_i \geq 0$. The payoff function of player $i$ is denoted as $\pi_i(x_1, ..., x_N) = p(\sum_{j=1}^{N} x_j)x_i-C_i(x_i)$, where $p(\cdot)$ is the market price function that maps the total production quantity to a price in $\mathbb{R}^{+}$, and $C_i(\cdot)$ is the production cost function of player $i$.
\end{definition}
The goal of each player in Cournot games is to decide the best production quantity $x_i \geq 0$ such that maximizes his payoff $\pi_i$. An important concept in game theory is the \emph{Nash equilibrium} (NE), at which state no player can increase his payoff via a unilateral deviation in action. The analysis of NE is important since it reveals a stable state of the multi-agent system: once the NE is reached, no one would have the economic incentive to break it. The NE of a Cournot game $(\pi_1, ..., \pi_N)$ is defined by a vector $\bm{{x}^{*}}$, such that $\forall i$,
\begin{equation}\label{eq:def_ne}
\pi_i(x_i^{*}, \bm{{x}_{-i}^{*}}) \geq \pi_i(x_i, \bm{{x}_{-i}^{*}}),  \forall x_i \geq 0\,,
\end{equation}
where $\bm{x}_{-i}$ denotes the actions of all players except $i$. The left side of Eq~\eqref{eq:def_ne} is player $i$'s payoff at the NE, and the right side is that of any deviated action given other players' actions fixed.  In this paper, we restrict our attentions to Cournot games that satisfy the following assumptions:
\begin{assumption}\label{assumption1}
	We assume the Cournot games satisfy:
	\begin{enumerate}
		\item[1)] The market price function $p$ is concave, strictly decreasing, and twice differentiable on $[0,\ymax]$, where $\ymax$ is the first point where $p$ becomes $0$. For $y>\ymax$, $p(y)=0$. In addition, $p(0)>0$. \hfill (A1)
		\item[2)] The individual cost function $C_i(x_i)$ is convex, strictly increasing, and twice differentiable, with $p(0)>C_i'(0)$, for all $i$. \hfill (A2)
	\end{enumerate}
\end{assumption}
These assumptions are standard in literature (e.g., see~\cite{johari2005efficiency} and the references within). The assumption $p(0)>C_i'(0)$ is to avoid the triviality of a player never participating in the game.~\cite{szidarovszky1977new} first proved that Cournot games with the above assumptions have unique Nash equilibrium.
\begin{proposition}\label{prop:unique_ne}
	A Cournot game satisfying (A1) and (A2) has exactly one Nash equilibrium.
\end{proposition}
The proof of Proposition~\ref{prop:unique_ne} is provided in Appendix B for readers' reference. Below, we briefly discuss two example applications of Cournot game in socio-economic systems.
\paragraph{Example 1 (Wholesale Electricity Market)}
The Cournot model is the most widely adopted framework for electricity market design~\cite{kirschen2004fundamentals}. Suppose there are $N$ electricity producers, each supplying the market with $x_i$ units of energy.
In an uncongested grid\footnote{In a congested grid, all electricity producers still compete in a Cournot game manner (i.e., bidding quantities), while the system operator that transmits electricity sets the congestion price to maximize social surplus of the entire system~\cite{yao2008modeling}.}, the electricity is priced as a decreasing function of the total generated electricity. For instance, consider both the market price and individual production cost function are linear,
the profit of generator $i$ can be written as: $\pi_i(x_i; \bm{x}_{-i}) = x_i (a-b\sum_{j=1}^{N} x_j) - c_i x_i,$
where $c_i\geq 0$ is the marginal production cost of $i$.
\paragraph{Example 2 (Lotteries)} Lotteries are becoming an increasingly important mechanism to allocate limited resources in social contexts, with examples in housing~\cite{friedman2014economics}, parking~\cite{zhang2015competition} and buying limited goods~\cite{phade2019optimal}. These lotteries typically allocate each player with a number of ``coupons'' and here we consider the coupon amount to be continuous (i.e., real numbers). The player's chance of winning depends on the number of coupons he owns, and the total number of coupons played in the round. Suppose $x_1, x_2, \dots,x_N$ are the coupons used by all players, then a decreasing price function $p(\sum_i x_i)$ can be used to model the fact that each player is less likely to win the lottery as others spend more coupons. The profit of player $i$ is
$\pi_i(x_i;\bm{x}_{-i})=p(\sum_{j=1}^N x_j) x_i-x_i$, where $-x_i$ represents the cost of spending the coupons.

\subsection{Review of No-Regret Algorithms}
The concept of Nash equilibrium is useful if players can reach it. However, in many practical settings, players do not have full information on the game (non-cooperative) and thus cannot pre-compute the NE beforehand. Thus the idea of ``learning'' the equilibrium arises and it becomes important to understand the dynamics of the iterative learning process. In this work, we focus on the class of learning algorithms with worst-case performance guarantees, namely the no-regret algorithms.

An algorithm is called no-regret~\cite{hazan2016introduction} (or no-external regret) if the difference between the total payoff it receives and that of the best-fixed decision in hindsight is sublinear as
a function of time. Formally, at each time step $t$, an online algorithm $\mathcal{A}$ selects an action vector $\mathbf{x}_t \in \mathcal{X}$. After $\mathbf{x}_t$ is selected, the algorithm receives $f_t$, and collects a payoff of $f_t(\mathbf{x}_t)$. All decisions are made online, in the sense that the algorithm does not know $f_{t}$ before choosing action $\mathbf{x}_t$.
Here all the payoff functions $f_1, f_2, ..., f_T \in \mathcal{F}$, where $\mathcal{F}$ is a bounded family of functions. Let $T$ denote the total number of game iterations. Then the total payoff collected by algorithm $\mathcal{A}$ until $T$ is $\sum_{t=1}^{T}f_t(\mathbf{x}_t)$, and the total payoff of a static feasible action $\mathbf{\tilde{x}}$ is $\sum_{t=1}^{T}f_t(\mathbf{\tilde{x}})$.
We formally define the regret of $\mathcal{A}$ after $T$ iterations as:
\begin{equation}
R_T(\mathcal{A}) = \max_{\mathbf{x} \in \mathcal{X}} \sum_{t=1}^{T} f_t(\mathbf{x}) - \sum_{t=1}^{T} f_t(\mathbf{x}_t) \,,
\end{equation}

An algorithm $\mathcal{A}$ is said to have no regret, if for every online sequential problem, $\{f_1, f_2, ..., f_T\} \subseteq \mathcal{F}$, the regret is sublinear as a function of $T$, i.e., $R_T(\mathcal{A}) = o(T)$. This implies that the algorithm performs (at least) as well as the best-fixed strategy in hindsight. Such a guarantee is desirable for rational players since no one wants to realize that the decision policy he/she employed is strictly inferior to the same action throughout.
There are a collection of algorithms satisfies the no-regret property, given the action set and the cost functions are both convex. Based on the information available to players, no-regret algorithms can be grouped into two types: zeroth-order (or bandit) algorithms and first-order (or gradient-based) algorithms.

\paragraph{Zeroth-order algorithms.} It accounts for extremely low information environments where players have only the realized payoff information, i.e., $f_t(\mathbf{x}_t)$ obtained from a given action $\mathbf{x}_t$ and nothing else. In game-theoretic settings (especially in non-cooperative games), the bandit feedback framework is more common since players usually only have local information and cannot tell with certainty what are the utilities and actions of other players.
The core of zeroth-order no-regret learning algorithms is to infer the gradient, i.e., obtaining an unbiased gradient estimator with bounded variance.
FKM~\cite{flaxman2005online} is a well-known zeroth-order no-regret algorithm under the single function evaluation situation, which is also known as ``gradient descent without a gradient''. The pseudocode of FKM is provided in Appendix A1 (Algorithm~\ref{alg:FKM}).

\paragraph{First-order algorithms.} As opposed to the zeroth-order algorithms, in first-order algorithms, an oracle that returns the payoff gradient at the queried action (i.e., $\nabla f_{t}(\mathbf{x}_t)$) is assumed available. Therefore, players can adjust their actions by taking a step towards the gradient direction, to maximize their utilities. Online mirror descent~\cite{shalev2007online} is a widely adopted first-order no-regret algorithm, which has been extensively studied under the learning in games setting~\cite{zhou2017mirror,zhou2018learning,mertikopoulos2019learning}. The pseudocode implementation of the online mirror descent algorithm is provided in Appendix A2 (Algorithm~\ref{alg:OMD}).

We want to emphasize that, the no-regret property only tells us about the time-average performance. From the players' perspective, they also care about (if not more) the performance of their final-iteration actions. However, deriving the final-iteration convergence based on the time-average regret is not easy. In this work, we prove both the time-average and final-iteration convergence of payoffs and actions, by ably using the Cournot game structure property, which we will discuss in more detail in Section~\ref{sec:theory}.

\subsection{Existing Convergence Results w.r.t. Cournot Game}\label{sec:existing_conv}
Existing learning in games literature mostly focus on the class of monotone games~\cite{rosen1965existence}.
\begin{definition}[Monotone game]\label{def:MC}
	A game is monotone (or so called diagonally strictly concave) is it satisfies, $\forall \mathbf{x}, \mathbf{x}' \in \mathcal{X}$
	\begin{equation}\label{eq:MC}
	\langle g(\mathbf{x}) - g(\mathbf{x}'), \mathbf{x} - \mathbf{x}' \rangle \leq 0\,,
	\end{equation}
	with equality if and only if $\mathbf{x} = \mathbf{x}'$, where $g(\mathbf{x})$ is the game gradient that $g(x) = \begin{bmatrix} \nabla_1 \pi_1(\mathbf{x}), \cdots, \nabla_N \pi_N(\mathbf{x})\end{bmatrix}^T$.
\end{definition}
Rosen~\cite{rosen1965existence} showed that every concave N-player game\footnote{A concave N-player game requires the action set to be convex and the individual payoff function to be concave in w.r.t. player's own action. Cournot games with assumptions (A1)-(A2) meet this definition.} satisfying this addittional monotonicity condition has a unique NE. He also showed that, starting from any feasible point in the action set, players will always converge to the NE, if they adapt their actions following the payoff gradients. In fact, Rosen's monotonicity condition is a common assumption and the \emph{cornerstone} for many convergence proofs in learning in games literature~\cite{bravo2018bandit}.

However, general Cournot games with assumptions (A1) and (A2) may not satisfy the monotone condition. See the counter example below.
\paragraph{Counter example}
Let consider a four-player Cournot game. The market price is a piecewise linear function with non-negative lower bound: $p(y) =\begin{cases}1-y & 0 \leq y \leq 1\\
0 & y>1\end{cases}$ and the individual production cost function is $C_i(x_i) = 0.05 x_i, \forall x_i \geq 0$. Hence, the payoff of each player is:
$$\pi_i(\mathbf{x}) = p(\sum_{j=1}^{4} x_j) x_i-0.05x_i, \forall i=1,2,3,4.$$
The payoff gradient is $\frac{\partial \pi_i(\mathbf{x})}{\partial x_i} = 0.95-\sum_{j=1}^{4} x_j-x_i\,,$ when $\sum_{j=1}^{4} x_j \leq 1$, and $\frac{\partial \pi_i(\mathbf{x})}{\partial x_i} = - 0.05$ otherwise.
Consider the following two points: $\mathbf{x} = \begin{bmatrix}0.2082, 0.2273, 0.1988, 0.2169\end{bmatrix}^{T}$ and
$\mathbf{x}' = \begin{bmatrix}0.3506,  0.3279, 0.0456, 0.4439\end{bmatrix}^{T}$.

It is easy to check that,
\begin{align*}\label{eq:vs_condition}
& \langle g(\mathbf{x}) - g(\mathbf{x}'), \mathbf{x}-\mathbf{x}' \rangle
= 0.0242 >0
\end{align*}
which contradicts the monotone game definition in Eq~\eqref{eq:MC}. The above counter example shows that the previously examined models and convergence results in monotone games do not apply to Cournot games. In fact, without this nice game structure assumption, it becomes much harder to analyze the dynamics and derive convergence results.

\section{Convergence Analysis in Cournot Games}
\label{sec:theory}
We discuss the main convergence results in this section. The first step is to select the right notion of convergence.
Next, we prove the convergence results in two steps, by first showing the payoff convergence, then deriving the action convergence. At the end of this section, we discuss the impact of different information and pricing mechanisms on the convergence rates.
\subsection{Convergence Definition}
\begin{definition}[Convergence in measure] \label{defn:measure}
	Let $\mu$ be a measure on $\mathbb{N}$.  We say that a sequence $a_t$ converges in measure to $a$ if
	$\forall \epsilon >0$, $\lim_{t \rightarrow \infty} \mu(|a_t - a| > \epsilon) = 0$.
\end{definition}
The reason we need to work with the notion of \emph{convergence in measure} rather than the standard notion of convergence (i.e., $\lim_{t \rightarrow \infty} a_t = a$) is that the latter condition is too stringent for no-regret algorithms. Consider the following example. Given a no-regret algorithm $\mathcal{A}$, we can construct another algorithm $\mathcal{A}'$ in the following manner. Let $M$ be some positive integer larger than $1$. Then the actions produced by $\mathcal{A}'$ is the same as $\mathcal{A}$ except for times $M$, $M^2$, $M^3$, $\dots$. At these times, $\mathcal{A}$' takes on the action $0$ (or any other arbitrary action). Both $\mathcal{A}$ and $\mathcal{A}'$ are no-regret algorithms, since $\mathcal{A}'$ only deviates at a set of vanishing small fraction of points. On the other hand, for $\mathcal{A}'$, its actions cannot converge in the standard sense. Therefore, given only the regret bound, the best final time convergence result we can hope for is convergence in measure as defined in Definition~\ref{defn:measure}.

\subsection{Payoff Convergence}
In this part, we prove the payoff convergence.  Theorem~\ref{thm:conv_payoff_avg} shows the time-average convergence and Theorem~\ref{thm:conv_payoff_actual} sharps the result by showing the final-iteration convergence.
\begin{theorem}[Time-average convergence]
	\label{thm:conv_payoff_avg}
	Suppose that after $T$ iterations, every player has expected regret $o(T)$.
	As $T \rightarrow \infty$, every player's time-average payoff $\frac{1}{T} \sum_{t=1}^{T} \pi_i(\mathbf{x}_t), \forall i$, converges to the payoff at the Nash equilibrium $\pi_i(\mathbf{x}^{*})$.
\end{theorem}
\begin{proof}
	Consider the $i$-th player. In each game iteration t, let $(x_{t, i}, \mathbf{x}_{t, -i})$ be the moves
	played by all the players.
	From player $i$’s point of view, the payoff he obtains at time $t$ is,
	\begin{equation}
	\forall \xi \in \mathcal{X}_i, \pi_i(\xi) = \pi_i(\xi, \mathbf{x}_{t, -i}).
	\end{equation}
	Note that this payoff function is concave with respective to his own action $\xi$ by assumption.

	By the definition of regret,
	\begin{equation}
	R_i(T) = \max_{\hat{x}_i \in \mathcal{X}_i} \sum_{t=1}^{T} \pi_i(\hat{x}_i, \mathbf{x}_{t, -i}) - \sum_{t=1}^{T} \pi_i(x_{t, i}, \mathbf{x}_{t, -i}).
	\end{equation}

	Equivalently, $\forall \hat{x}_i \in \mathcal{X}_i$,
	\begin{equation}\label{eq:thm1}
	\frac{1}{T} \sum_{t=1}^{T} \pi_i(x_{t, i}, \mathbf{x}_{t, -i}) \geq \frac{1}{T} \sum_{t=1}^{T} \pi_i(\hat{x}_i, \mathbf{x}_{t, -i}) - \frac{R_i(T)}{T} .
	\end{equation}

	Let consider the best response of player $i$ at time $t$ given all other players' actions as $(x_{t, i}^{*}) = \arg \max_{\xi} \pi_i({\xi, \mathbf{x}_{t, -i}})$. Obviously, player $i$'s payoff is upper bounded by his best response payoff by definition,
	\begin{equation}\label{eq:thm2}
	\frac{1}{T} \sum_{t=1}^{T} \pi_i(x_{t, i}, \mathbf{x}_{t, -i}) \leq \frac{1}{T} \sum_{i=1}^{T} \pi_{i}(x_{t, i}^{*}, \mathbf{x}_{t, -i}).
	\end{equation}

	In addition, since $\pi_i$ is concave with respect to $x_i$, it follows:
	\begin{align}\label{eq:thm3}
	\frac{1}{T} \sum_{i=1}^{T} \pi_{i}(x_{t, i}^{*}, \mathbf{x}_{t, -i}) & \leq
	\frac{1}{T} \sum_{t=1}^{T} \pi_1(\tilde{x}_i, \mathbf{x}_{t, -i})
	\end{align}
	where $\tilde{x}_i = \frac{\sum_{t=1}^{T} x_{t, i}^{*}}{T}$.

	Combining Eq.~\eqref{eq:thm1} and~\eqref{eq:thm3} we have, the difference between the actual payoff and best response is bounded by the regret,
	\begin{align}\label{eq:no_regret_eq3}
	\frac{1}{T} \sum_{t=1}^{T} \pi_i(x_{t, i}, \mathbf{x}_{t, -i})
	\geq \frac{1}{T} \sum_{i=1}^{T} \pi_{i}(x_{t, i}^{*}, \mathbf{x}_{t, -i}) - \frac{R_i(T)}{T}.
	\end{align}

	Combine the lower bound in~\eqref{eq:no_regret_eq3} and upper bound in~\eqref{eq:thm2},
	\begin{align}\label{eq:squeeze}
	\frac{1}{T} \sum_{i=1}^{T} \pi_{i}(x_{t, i}^{*}, \mathbf{x}_{t, -i}) - \frac{R_i(T)}{T} &\leq \frac{1}{T} \sum_{t=1}^{T} \pi_i(x_{t, i}, \mathbf{x}_{t, -i}) \,,\nonumber\\
	& \leq \frac{1}{T} \sum_{t=1}^{T} \pi_i(x_{t, i}^{*}, \mathbf{x}_{t, -i}).
	\end{align}

	Since we know that $R_i(T) = o(T)$ as $T \rightarrow \infty$, use the Squeeze theorem in calculus,
	\begin{equation}\label{eq:time_avg_conv}
	\lim_{T \rightarrow \infty} \frac{1}{T} \sum_{t=1}^{T} \pi_i(x_{t, i}, \mathbf{x}_{t, -i}) = \lim_{T \rightarrow \infty} \frac{1}{T} \sum_{t=1}^{T} \pi_i(x_{t, i}^{*}, \mathbf{x}_{t, -i})\,,
	\end{equation}
	which holds for all players. Therefore, as $T \rightarrow \infty$, the average payoff of each player converges to the payoff at his best response. As every player plays his best response against the other players \emph{simultaneously}, the time-average payoff converges to the Nash equilibrium.
\end{proof}
\begin{theorem}[Final-iteration convergence]
	\label{thm:conv_payoff_actual}
	Suppose that after $T$ iterations, every player has expected regret $o(T)$.
	As $T \rightarrow \infty$, every player's actual payoff $\pi_i(\mathbf{x}_t), \forall i$, converges to the payoff at the Nash equilibrium $\pi_i(\mathbf{x}^{*})$ in measure,
	$$\forall \epsilon >0, \lim_{t \rightarrow \infty} \mu(|\pi_i(\mathbf{x}_t) - \pi_i(\mathbf{x}^{*})| > \epsilon) = 0.$$
\end{theorem}
\begin{proof}
	We prove Theorem~\ref{thm:conv_payoff_actual} by contradition. In particular, suppose that $\exists \epsilon>0$, such that more than a sub-linear fraction of $t \in \{1, 2, ..., T\}$ satisfies that: $|\pi_i(\mathbf{x}_t) - \pi_i(\mathbf{x}^{*})| > \epsilon$.

	Let define the following notations for the proof. Denote $a_t = \pi_{i}(x_{t, i}^{*}, \mathbf{x}_{t, -i})$, which is the best response payoff for player $i$ given others' action, and $b_t = \pi_i(\mathbf{x}_t)$. Thus
	\begin{align}
	0 \leq b_t \leq a_t\,,
	\end{align}
	Now, let re-arrange all the time steps such that the time where $|b_t-a_t| > \epsilon$ show up in the front. Say there are $T_1$ such points, then
	\begin{align}\label{eq:ref2}
	&\lim_{T \rightarrow \infty} |\frac{1}{T} \sum_{t=1}^{T} (b_t-a_t)|
	= \lim_{T \rightarrow \infty} \frac{1}{T} \sum_{t=1}^{T} |b_t-a_t| (b_t \leq a_t) \,,\nonumber\\
	&= \lim_{T \rightarrow \infty} (\frac{1}{T} \sum_{t=1}^{T_1} |b_t-a_t|+\frac{1}{T} \sum_{t=T_1}^{T} |b_t-a_t|) \nonumber\\
	&\geq \lim_{T \rightarrow \infty} (\frac{T_1}{T} \epsilon + \frac{1}{T} \sum_{t=T_1}^{T} |b_t-a_t|)\,,
	\end{align}
	Since $T_1$ accounts for more than a sub-linear fraction of $T$,  we have $\frac{T_1}{T} \nrightarrow 0$ as $T \rightarrow \infty$. Following Eq. \eqref{eq:ref2},
	\begin{align}\label{eq:ref3}
	\lim_{T \rightarrow \infty} |\frac{1}{T} \sum_{t=1}^{T} (b_t-a_t)|
	&\geq (\lim_{T \rightarrow \infty} \frac{T_1}{T}) \cdot \epsilon > o(T)
	\end{align}
	which contradicts the definition of no-regret algorithms.

	Hence, given any $\epsilon > 0$, as $T \rightarrow \infty$ there exists at most a \emph{measure zero} set of time such that $|b_t-a_t| > \epsilon$.
	Since this holds for all players simultaneously, we have,
	\begin{align}
	\forall i, \lim_{T \rightarrow \infty} \pi_i(\mathbf{x}_t) = \pi_i(\mathbf{x}^{*})\,, \forall t \in [1, ..., T]
	\end{align}
	for all but a measure zero set of time.
\end{proof}
One can interpret Theorem~\ref{thm:conv_payoff_actual} from two angles. On the one hand, given any $\epsilon >0$ (fix the error bound), the set of time that the actual payoff significantly deviates from the NE payoff equals $\frac{R_i(T)}{T}$. For no-regret algorithms with tighter regret bound, the set of time that far away from NE vanishes faster. On the other hand, after $T$ time steps (fix the number of iterations), we have that $\forall i, |\pi_i(\mathbf{x}_t) - \pi_i(\mathbf{x}^{*})| < O(\frac{R_i(T)}{T})$, for all $t \in T$ but a measure zero set. Thus, after the same number of iterations, algorithms with tighter regret bound have smaller error bound.

\subsection{Action Convergence}
Now we turn our attention to prove the action convergence. The following two propositions are needed for the proof.
\begin{proposition}[Inverse function theorem~\cite{stromberg2015introduction}]
	Consider function $f: \mathbb{R}^n \rightarrow \mathbb{R}^n$, and $f(\mathbf{x}_0) = \mathbf{y}_0$. Let $J = \frac{\partial f}{\partial \mathbf{x}}|_{\mathbf{x} = \mathbf{x}_0}$ as the Jacobian of function $f$. If $J$ evaluated at $\mathbf{x}_0$ is invertible, then there exists a continuous and differentiable function $g$ such that,
	$$g(f(\mathbf{x})) = \mathbf{x}\,,$$
	for $\mathbf{x} \in \mathcal{X}$ and $\mathbf{y} \in \mathcal{Y}$ where $\mathcal{X}$ is some open set around $\mathbf{x}_0$ and $\mathcal{Y}$ is some open set around $\mathbf{y}_0$.
\end{proposition}

\begin{proposition}[Lipschitz continuity]
	A function $f$ from $\mathcal{X} \subset \mathbb{R}^n$ into $\mathbb{R}^n$ is Lipschitz continuous at $\mathbf{x_1} \in \mathcal{X}$ if there is constant $L \in \mathbb{R}^{+}$ such that,
	$$||f(\mathbf{x}_2)-f(\mathbf{x}_1)||_2 \leq L||\mathbf{x}_2-\mathbf{x}_1||_2\,,$$
	for all $\mathbf{x}_2 \in \mathcal{X}$ sufficiently near $\mathbf{x}_1$.
\end{proposition}

\begin{theorem}[Convergence in action]
	\label{thm:action_conv}
	Let $\mathbf{x}^{*}$ denote the Nash equilibrium, suppose that $\mathbf{x}$ satisfies:
	$$||\pi(\mathbf{x})- \pi(\mathbf{x}^{*})||_2 \leq \epsilon\,, \text{(closeness in payoff)}\,, $$
	then it implies that,
	$$||\mathbf{x} - \mathbf{x}^{*}||_2 \leq L \cdot \epsilon\,, \text{(closeness in action)}\,, $$
	where $\pi(\mathbf{x}) = [\pi_1(\mathbf{x}), ..., \pi_N(\mathbf{x})]^T$ is the payoff function/vector of a N-player Cournot game with assumptions (A1) and (A2), and $L \in \mathbb{R}^{+}$ is a constant.
\end{theorem}
\begin{proof}
	Recall that the individual payoff function in Cournot games is, $\pi_i(\mathbf{x}) = p(\sum_{j=1}^{N} \mathbf{x}_j)x_i - C_i(x_i)$, and $\pi(\mathbf{x}) = [\pi_1(\mathbf{x}), ..., \pi_N(\mathbf{x})]^T$ is the collection of all players' payoffs.

	Let $J = \frac{\partial \pi}{\partial \mathbf{x}}$ denotes the Jacobian of function $\pi(\mathbf{x})$. Firstly, we show that $J(\mathbf{x}^{*})$ is non-singular, where $\mathbf{x}^{*}$ is the NE. For the Jacobian entries, we have
	$J_{i, i}(\mathbf{x}^{*}) = \frac{\partial \pi_i(\mathbf{x})}{\partial x_i} = 0$ (diagonal entries) and $J_{i, j (i \neq j)}(\mathbf{x}^{*}) = \frac{\partial \pi_i(\mathbf{x})}{\partial x_j} = p'(\sum_{j=1}^{N} \mathbf{x}_j^{*}) x_i^{*}$ (non-diagonal entries). Then the Jacobian equals,
	\begin{align}
	J(\mathbf{x}^{*}) =
	\begin{bmatrix}
	0 & P^{*} x_1^{*} &  \cdots & P^{*} x_1^{*}\\
	P^{*} x_2^{*} & 0 & \cdots & P^{*} x_2^{*}\\
	& & \ddots & \\
	P^{*} x_N^{*} & P^{*} x_N^{*} & \cdots & 0
	\end{bmatrix}.
	\end{align}
	where $P^{*} = p'(\sum_{j=1}^{N} \mathbf{x}_j^{*})$ is the market price at the NE.

	Concisely, $J(\mathbf{x}^{*})$ can be written as,
	\begin{align}
	J(\mathbf{x}^{*}) &= P^{*}(\mathbf{x}^{*} \cdot \mathbf{1}^T) - P^{*} \cdot diag(\mathbf{x}_1^{*}, \mathbf{x}_2^{*}..., \mathbf{x}_N^{*}).
	\end{align}

	We argue that $P^{*} > 0$ (market price at the NE is positive). Suppose that at the NE, $\sum_{i} x_i^{*} \geq y_{max}$ where $y_{max}$ is the first point such that $p$ becomes zero. Then at least one of the $x_i^{*}$ is positive, as by assumption $p(0)>0$. However under this case, firm $i$ can be \emph{strictly better} off if it reduces $x_i^{*}$, which contradicts the definition of NE. Therefore, we have $P^{*} > 0$ and $\sum_{i} x_i^{*} < y_{max}$ at the NE.

	Thus, in order to show $J(\mathbf{x}^{*})$ is invertible, it suffices to show that $\left(\mathbf{x}^{*} \cdot \mathbf{1}^T -  diag(\mathbf{x}_1^{*}, \mathbf{x}_2^{*}..., \mathbf{x}_N^{*})\right)$ is invertible. Suppose $\mathbf{v} = [v_1, v_2, ..., v_N]^T$ solves the following equation,
	\begin{equation}\label{eq: jacobian_v_zero}
	\left(\mathbf{x}^{*} \cdot \mathbf{1}^T -  diag(\mathbf{x}_1^{*}, \mathbf{x}_2^{*}..., \mathbf{x}_N^{*})\right) \mathbf{v} = 0\,,
	\end{equation}
	Since $\forall i, x_i^{*} \neq 0$ (game admits no trivial solutions), the above linear system has the same solution as the following,
	\begin{equation}
	(\mathbf{1}^T \mathbf{v})  \mathbf{1} - \mathbf{v} = 0.
	\end{equation}
	which holds iff $\mathbf{v} = 0$.

	Therefore, $\left(\mathbf{x}^{*} \cdot \mathbf{1}^T -  diag(\mathbf{x}_1^{*}, \mathbf{x}_2^{*}..., \mathbf{x}_N^{*})\right)$ is invertible, and it follows $J(\mathbf{x}^{*})$ is also invertible.

	By the inverse function theorem, as $J(\mathbf{x}^{*})=\frac{\partial \pi}{\partial \mathbf{x}}|_{\mathbf{x}=\mathbf{x}^{*}}$ is invertible, there exists a continuously differentiable function $g$ (as the inverse function of $\pi$) such that,
	\begin{align}
	g(\pi(\mathbf{x})) = \mathbf{x}\,, \forall \pi \in \{\hat{\pi} \in \mathbb{R}^n: ||\hat{\pi}-\pi^{*}||\leq \epsilon \}\,,
	\end{align}

	By Lipschitz continuity, we have,
	\begin{align}
	||\mathbf{x}-\mathbf{x}^{*}||_2 &= ||g(\pi(\mathbf{x})) - g(\pi(\mathbf{x}^{*}))||_2 \,,\nonumber\\
	&\leq L ||\pi(\mathbf{x})- \pi(\mathbf{x}^{*})||_2
	\end{align}
	Therefore, given the payoffs are close, i.e., $||\pi(\mathbf{x})- \pi(\mathbf{x}^{*})||_2 \leq \epsilon$, the actions will also be close, i.e., $||\mathbf{x}-\mathbf{x}^{*}||_2 \leq L \cdot \epsilon$ and $L \in \mathbb{R}^{+}$.
\end{proof}
\subsection{Convergence Rate}
The previous parts proved that the payoffs and actions both converge to the NE of the game, under no-regret dynamics. However, the derivation steps do not explicitly provide us the convergence rate. In this part, we complete the analysis by discussing the convergence rate. Recall that the payoff convergence is based on the following equation (Eq~\eqref{eq:squeeze} in Section 3.2),
\begin{align*}
\frac{1}{T} \sum_{i=1}^{T} \pi_{i}(x_{t, i}^{*}, \mathbf{x}_{t, -i}) - \frac{R_i(T)}{T} & \leq \frac{1}{T} \sum_{t=1}^{T} \pi_i(x_{t, i}, \mathbf{x}_{t, -i})\\
& \leq \frac{1}{T} \sum_{t=1}^{T} \pi_i(x_{t, i}^{*}, \mathbf{x}_{t, -i})\,,
\end{align*}
where $R_i(T)$ is the algorithm regret after $T$ iterations. Therefore, the rate of convergence naturally connects to the algorithm's regret bound.
\paragraph{Zeroth-order algorithm} ~\cite{flaxman2005online} provides that the regret bound of FKM algorithm is $R(T) = O(T^{\frac{3}{4}})$. By Theorem~\ref{thm:conv_payoff_actual}, we have that,
\begin{align*}
||\pi(\mathbf{x}_{t}) - \pi(\mathbf{x^{*}})||_2 \leq O(\frac{R(T)}{T}) = O(T^{-\frac{1}{4}})\,,
\end{align*}
for all $t$ (except a measure zero set of time steps), as $T$ goes to infinity. Since the payoff is bounded, by Theorem~\ref{thm:action_conv}, the action is also bounded $||\mathbf{x}_t-\mathbf{x}^{*}||_2 \leq O(T^{-\frac{1}{4}})$
for all but a measure zero set of time.
\paragraph{First-order algorithm} The regret bound for online mirror descent (OMD) is $R(T) = O(T^{\frac{1}{2}})$~\cite{hazan2016introduction}. Similarly, we have that,
\begin{align*}
||\pi(\mathbf{x}_t) - \pi(\mathbf{x^{*}})||_2 \leq O(\frac{R(T)}{T}) = O(T^{-\frac{1}{2}})\,,
\end{align*}
for all $t$ but a measure zero set, as $T$ goes to infinity. By Theorem 3, the action is also bounded by
$||\mathbf{x}_t-\mathbf{x}^{*}||_2 \leq O(T^{-\frac{1}{2}})$
for all but a set of measure zero time.

Comparing the convergence rates between zeroth-order algorithms and first-order algorithms, we find that the benefits of having access to the gradient information are $O(T^{-\frac{1}{4}})$ in terms of player's equilibration rate. Using this insight, it is interesting to think from the market operator's shoe. In most current markets (e.g., electricity market), the system operator only provides the zeroth-order information for participants. However, our results suggest that by offering more information, the market (aggregate production levels, prices) can converge to the stable state faster. This observation provides a new angle to the vast amount of economics literature (e.g.,~\cite{athey2018value} and the references within) on studying the value of information in game efficiency. Our results imply that sharing more information can not only improve market efficiency but will also contribute to better computational performance.

\subsection{Discussion on Game Structure and Convergence}
Finally, we discuss how the convergence rate is affected by the game structure. In section~\ref{sec:existing_conv}, we gave a counter example showing that Cournot games may not be monotone games. But what if we restrict the price and individual cost function class such that Cournot games satisfy the monotonicity property? Will it lead to different convergence rates?

For example, consider a Cournot game with linear price function $P(\sum_{i=1}^{N} x_i) = 1-\sum_{i=1}^{N} x_i$ and linear individual cost $C_i(x_i) = x_i, \forall i$. By simple calculation, we find that this game is not only monotone but also strongly monotone in a sense that,
\begin{equation}
\sum_{i \in N} \lambda_i \langle g_i(\mathbf{x}')-g_i(\mathbf{x}), \mathbf{x}'-\mathbf{x} \rangle \leq -\frac{\beta}{2} ||\mathbf{x}'-\mathbf{x}||^2\,,
\end{equation}
for some $\lambda_i, \beta > 0$ and for all $\mathbf{x}, \mathbf{x'} \in \mathcal{X}$.~\cite{bravo2018bandit} proved that in strongly monotone games, zeroth-order no-regret algorithms can achieve $O(T^{-\frac{1}{3}})$ convergence rate, and first-order algorithms have $O(T^{-1})$ convergence rate.

Compared to our results in general Cournot games, that is $O(T^{-\frac{1}{4}})$ for zeroth-order algorithms and $O(T^{-\frac{1}{2}})$ for first-order algorithms, the benefits of using linear price function and having the strongly monotone property can be measured quantitatively. This provides yet another useful insight for market designers on the impact of price function design (hence the game property) on the market equilibration rate, in addition to the information mechanism design.

\section{Numerical Experiments}
\label{sec:results}
We provide two Cournot game examples and visualize the no-regret dynamics in these games. These toy examples aim to help readers quickly grasp the key theoretic results from three perspectives: 1) the convergence behavior; 2) the convergence rate differences between zeroth-order and first-order no-regret algorithms; and 3) the impact of game structure on convergence rates.

\noindent \textbf{Setup} We consider two four-player Cournot games with different market price and individual cost functions. $\textbf{G1:}$ a monotone Cournot game where $p(\bm{x}) = 1-(\sum_i x_i)$, and a linear production cost function is $C_i(x_i)= 0.05 x_i$. $\textbf{G2:}$ a Cournot game that is not monotone. We take the counter example in Section~\ref{sec:existing_conv} that the price function is piecewise linear that $p(\bm{x}) = 1-(\sum_{i}x_i), 0 \leq \sum_{i}x_i \leq 1$ and $p(\bm{x})=0$ otherwise. The individual production cost is $C_i(x_i) = 0.05 x_i, \forall x_i \geq 0$. The Nash equilibrium for both games is $x_1^{*} = x_2^{*} = x_3^{*} = x_4^{*} = 0.19$, leveraging the fact we proved within Theorem~\ref{thm:action_conv} that $\sum_i x_i^{*} \leq y_{max}$ for G2.

Both games proceed as follows. At each time step, every player simultaneously picks a production level, and then the market price is determined by their joint production and broadcasted back to all players. Each player calculates his own payoff by multiplying the production level by the market price, subtracting the cost. According to the observed payoff, players adjust their action strategies for the next round.
The game is repeated for multiple times with players either all using the zeroth-order FKM algorithm or the first-order OMD algorithm. The algorithm implementation details can be found in Appendix A3. In the OMD case, each player's payoff gradient is also calculated and broadcasted back to the corresponding player. We record the actions, payoffs, and the market price at each round.
\begin{figure}[htbp]
	\vskip 0.2in
	\begin{center}
		\subfigure[G1: FKM]{\includegraphics[width = 0.45\columnwidth]{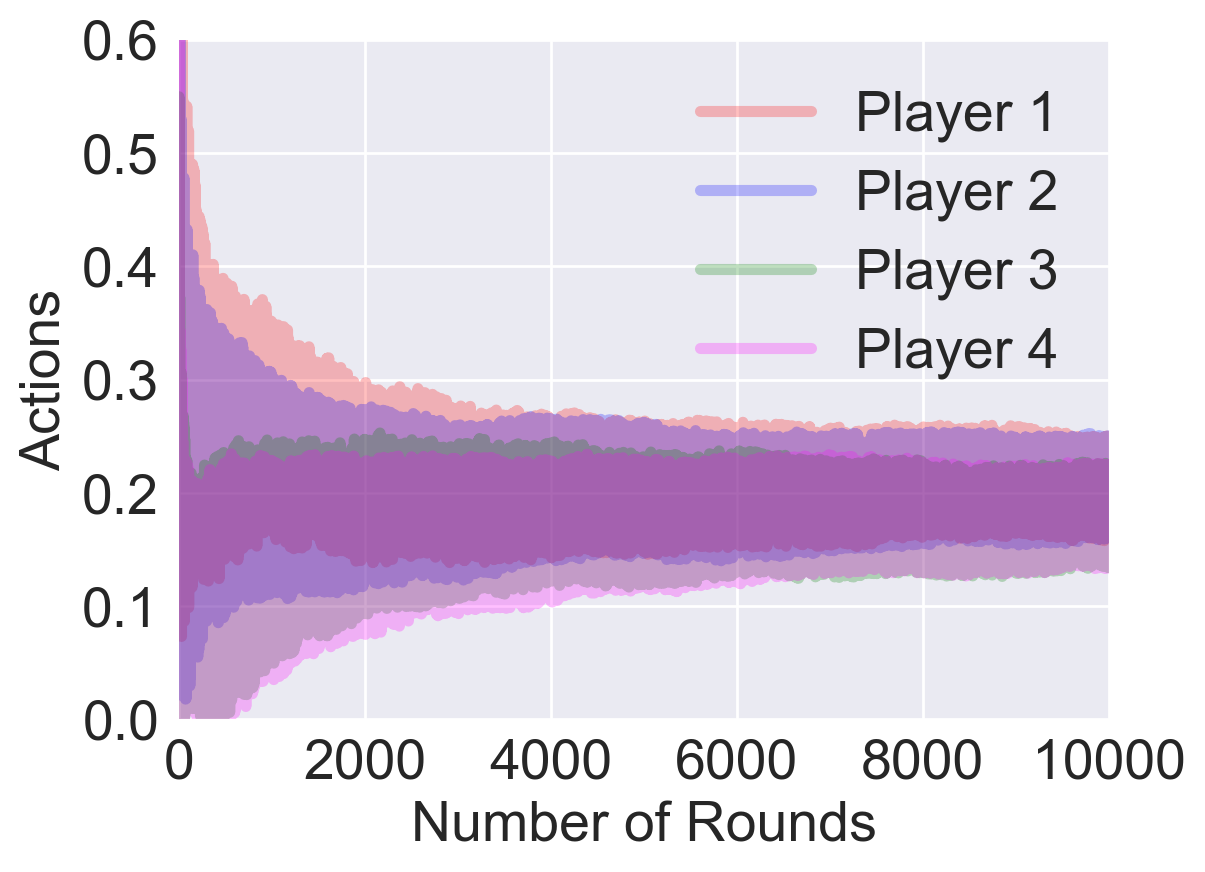}}
		\subfigure[G1: OMD]{\includegraphics[width = 0.45\columnwidth]{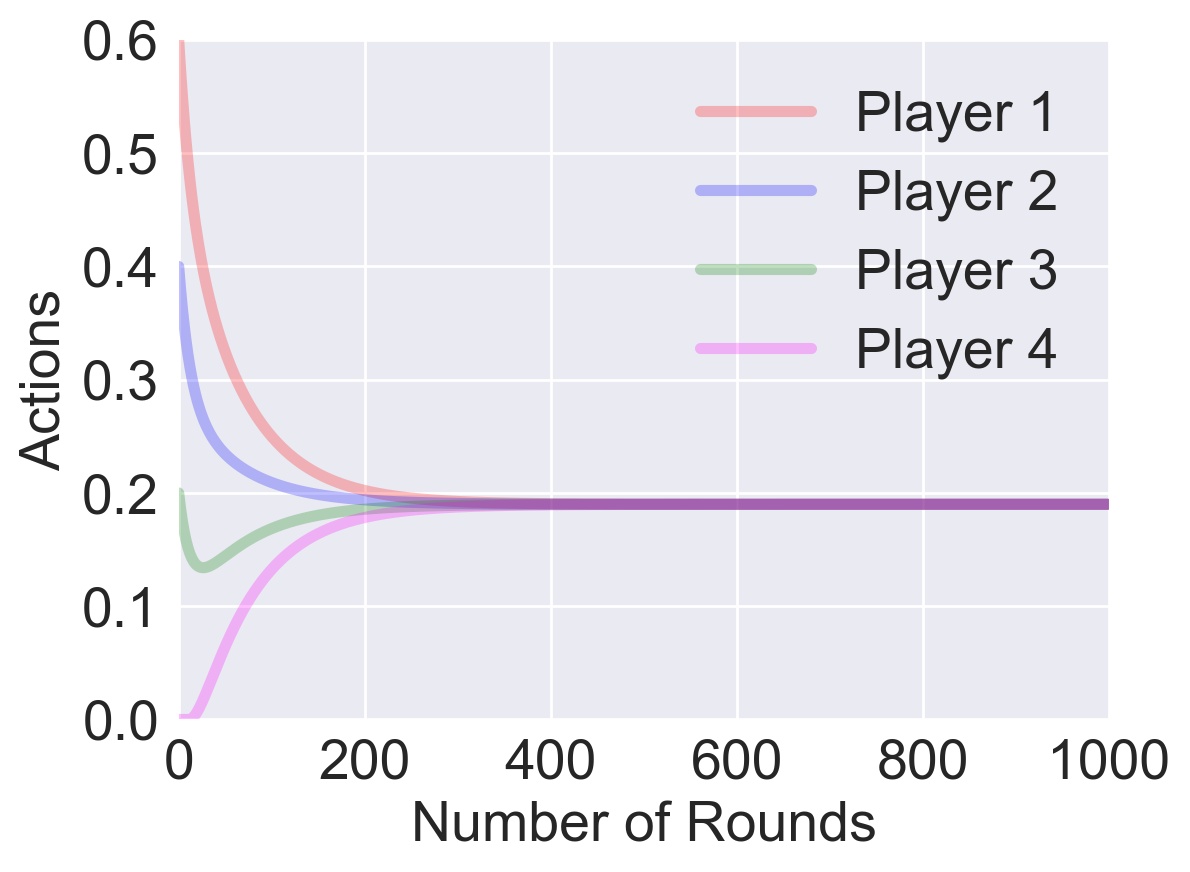}}\\
		\subfigure[G2: FKM]{\includegraphics[width = 0.45\columnwidth]{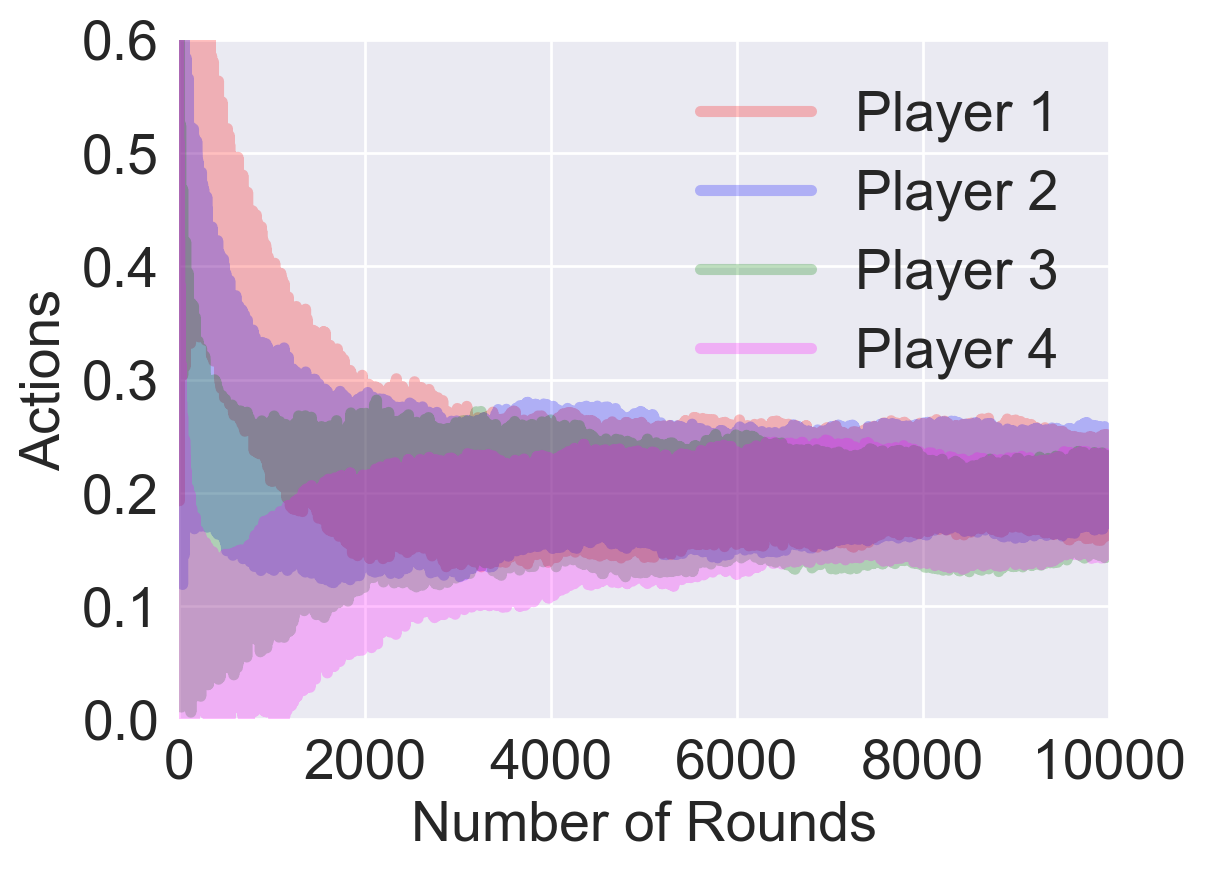}}
		\subfigure[G2: OMD]{\includegraphics[width = 0.45\columnwidth]{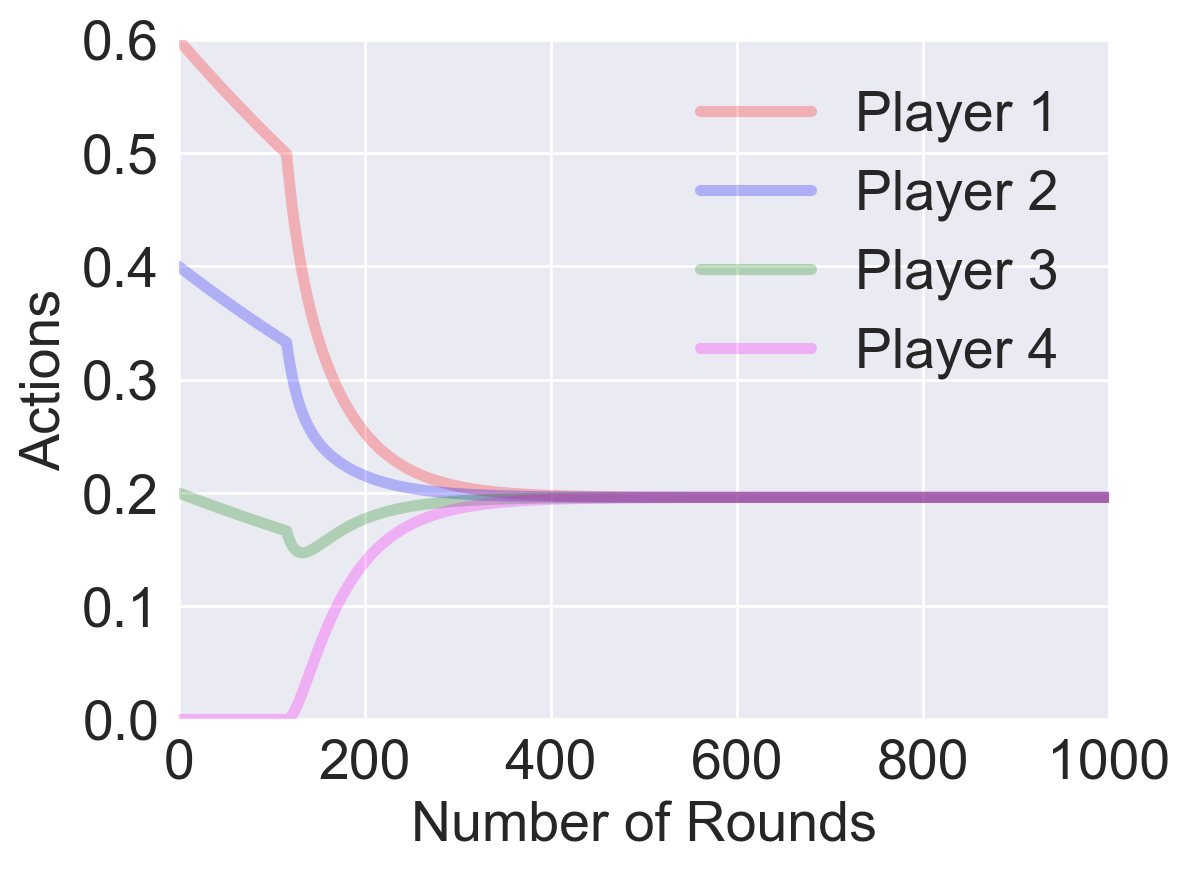}}
		\caption{Convergence behavior of FKM and OMD in two example Cournot games. }
		\label{fig:conv}
	\end{center}
	\vskip -0.2in
\end{figure}

Fig.~\ref{fig:conv} shows that all players' actions converge to the NE under both games, for both FKM and OMD algorithms. However, the convergence rate differs significantly. Comparing the performance of FKM (left column) against the performance of OMD (right column), it is obvious that the convergence rate of OMD is much faster, which demonstrates the benefits of having access to the payoff gradient information.
In addition, viewing the convergence rate difference between $G1$ (upper row) and $G2$ (bottom row), we find that the convergence is faster in $G1$ when the game is monotone, for both algorithms. This illustrates the gain of certain game structure. These observations are aligned with the theoretical results in Section~\ref{sec:theory}. 
\section{Conclusion}
In this paper, we study the interaction of strategic players in Cournot games with concave price functions and convex cost functions. We consider the dynamics of players actions and payoffs when all players use no-regret algorithms. We prove the time-average and final-iteration convergence for both payoffs and actions. Furthermore, we quantify the value of information and game structure in terms of players' convergence rates. It suggests that  different information and game structures can lead to faster convergence rate, which provides insights for market mechanism design with Cournot models. Our work is a strict generalization of all previously examined models, as they apply to \emph{all} no-regret dynamics in general Cournot games, without game structure assumptions and restrictions on the type of algorithms.
\bibliography{reference}

\begin{thebibliography}{37}
\providecommand{\natexlab}[1]{#1}
\providecommand{\url}[1]{\texttt{#1}}
\expandafter\ifx\csname urlstyle\endcsname\relax
  \providecommand{\doi}[1]{doi: #1}\else
  \providecommand{\doi}{doi: \begingroup \urlstyle{rm}\Url}\fi

\bibitem[Athey \& Levin(2018)Athey and Levin]{athey2018value}
Athey, S. and Levin, J.
\newblock The value of information in monotone decision problems.
\newblock \emph{Research in Economics}, 72\penalty0 (1):\penalty0 101--116,
  2018.

\bibitem[Barman \& Ligett(2015)Barman and Ligett]{barman2015finding}
Barman, S. and Ligett, K.
\newblock Finding any nontrivial coarse correlated equilibrium is hard.
\newblock \emph{ACM SIGecom Exchanges}, 14\penalty0 (1):\penalty0 76--79, 2015.

\bibitem[Bervoets et~al.(2016)Bervoets, Bravo, and Faure]{bervoets2016learning}
Bervoets, S., Bravo, M., and Faure, M.
\newblock Learning and convergence to nash in games with continuous action
  sets.
\newblock Technical report, Working paper, 2016.

\bibitem[Bimpikis et~al.(2014)Bimpikis, Ehsani, and Ilkili\c{c}]{Kostas14ec}
Bimpikis, K., Ehsani, S., and Ilkili\c{c}, R.
\newblock Cournot competition in networked markets.
\newblock In \emph{Proceedings of the 15th ACM Conference on Economics and
  Computation (EC 2014)}, pp.\  733, 2014.

\bibitem[Bravo et~al.(2018)Bravo, Leslie, and Mertikopoulos]{bravo2018bandit}
Bravo, M., Leslie, D., and Mertikopoulos, P.
\newblock Bandit learning in concave n-person games.
\newblock In \emph{Proceedings of the 32nd International Conference on Neural
  Information Processing Systems (NeurIPS 2018)}, pp.\  5661--5671, 2018.

\bibitem[Cai et~al.(2019)Cai, Bose, and Wierman]{cai2019role}
Cai, D., Bose, S., and Wierman, A.
\newblock On the role of a market maker in networked cournot competition.
\newblock \emph{Mathematics of Operations Research}, 44\penalty0 (3):\penalty0
  1122--1144, 2019.

\bibitem[Chletsos \& Saiti(2019)Chletsos and Saiti]{chletsos2019hospitals}
Chletsos, M. and Saiti, A.
\newblock Hospitals as suppliers of healthcare services.
\newblock In \emph{Strategic Management and Economics in Health Care}, pp.\
  179--205. Springer, 2019.

\bibitem[Cohen et~al.(2017)Cohen, H{\'e}liou, and
  Mertikopoulos]{cohen2017learning}
Cohen, J., H{\'e}liou, A., and Mertikopoulos, P.
\newblock Learning with bandit feedback in potential games.
\newblock In \emph{Proceedings of the 31st International Conference on Neural
  Information Processing Systems (NeurIPS 2017)}, pp.\  6372--81, 2017.

\bibitem[Cournot(1838)]{cournot1838recherches}
Cournot, A.-A.
\newblock \emph{Recherches sur les principes math{\'e}matiques de la
  th{\'e}orie des richesses par Augustin Cournot}.
\newblock chez L. Hachette, 1838.

\bibitem[Flaxman et~al.(2005)Flaxman, Kalai, Kalai, and
  McMahan]{flaxman2005online}
Flaxman, A.~D., Kalai, A.~T., Kalai, A.~T., and McMahan, H.~B.
\newblock Online convex optimization in the bandit setting: gradient descent
  without a gradient.
\newblock In \emph{Proceedings of the 16th annual ACM-SIAM symposium on
  Discrete algorithms (SODA 2005)}, pp.\  385--394, 2005.

\bibitem[Foster et~al.(2016)Foster, Li, Lykouris, Sridharan, and
  Tardos]{foster2016learning}
Foster, D.~J., Li, Z., Lykouris, T., Sridharan, K., and Tardos, E.
\newblock Learning in games: Robustness of fast convergence.
\newblock In \emph{Proceedings of the 30th International Conference on Neural
  Information Processing Systems (NeurIPS 2016)}, pp.\  4734--4742, 2016.

\bibitem[Friedman \& Weinberg(2014)Friedman and
  Weinberg]{friedman2014economics}
Friedman, J.~H. and Weinberg, D.~H.
\newblock \emph{The economics of housing vouchers}.
\newblock Academic Press, 2014.

\bibitem[Gordon(2007)]{gordon2007no}
Gordon, G.~J.
\newblock No-regret algorithms for online convex programs.
\newblock In \emph{Proceedings of the 21st International Conference on Neural
  Information Processing Systems (NeurIPS 2007)}, pp.\  489--496, 2007.

\bibitem[Hazan et~al.(2016)]{hazan2016introduction}
Hazan, E. et~al.
\newblock \emph{Introduction to online convex optimization}.
\newblock Foundations and Trends in Optimization., 2016.

\bibitem[Ito \& Reguant(2016)Ito and Reguant]{ito2016sequential}
Ito, K. and Reguant, M.
\newblock Sequential markets, market power, and arbitrage.
\newblock \emph{American Economic Review}, 106\penalty0 (7):\penalty0 1921--57,
  2016.

\bibitem[Johari \& Tsitsiklis(2005)Johari and Tsitsiklis]{johari2005efficiency}
Johari, R. and Tsitsiklis, J.~N.
\newblock Efficiency loss in cournot games.
\newblock \emph{Harvard University}, 2005.

\bibitem[Kirschen \& Strbac(2004)Kirschen and Strbac]{kirschen2004fundamentals}
Kirschen, D.~S. and Strbac, G.
\newblock \emph{Fundamentals of power system economics}, volume~1.
\newblock Wiley Online Library, 2004.

\bibitem[Lanctot et~al.(2017)Lanctot, Zambaldi, Gruslys, Lazaridou, Tuyls,
  P{\'e}rolat, Silver, and Graepel]{lanctot2017unified}
Lanctot, M., Zambaldi, V., Gruslys, A., Lazaridou, A., Tuyls, K., P{\'e}rolat,
  J., Silver, D., and Graepel, T.
\newblock A unified game-theoretic approach to multiagent reinforcement
  learning.
\newblock In \emph{Proceedings of the 31st International Conference on Neural
  FInformation Processing Systems (NeurIPS 2017)}, pp.\  4193--4206, 2017.

\bibitem[Li et~al.(2019)Li, Carboni, Gonzalez, Campolo, and
  Burdet]{li2019differential}
Li, Y., Carboni, G., Gonzalez, F., Campolo, D., and Burdet, E.
\newblock Differential game theory for versatile physical human--robot
  interaction.
\newblock \emph{Nature Machine Intelligence}, 1\penalty0 (1):\penalty0 36--43,
  2019.

\bibitem[Marden \& Shamma(2015)Marden and Shamma]{marden2015game}
Marden, J.~R. and Shamma, J.~S.
\newblock Game theory and distributed control.
\newblock In \emph{Handbook of game theory with economic applications}, pp.\
  861--899. Elsevier, 2015.

\bibitem[Marden \& Wierman(2009)Marden and Wierman]{marden2009overcoming}
Marden, J.~R. and Wierman, A.
\newblock Overcoming limitations of game-theoretic distributed control.
\newblock In \emph{Proceedings of the 48h IEEE Conference on Decision and
  Control (CDC 2009)}, pp.\  6466--6471, 2009.

\bibitem[Mertikopoulos \& Zhou(2019)Mertikopoulos and
  Zhou]{mertikopoulos2019learning}
Mertikopoulos, P. and Zhou, Z.
\newblock Learning in games with continuous action sets and unknown payoff
  functions.
\newblock \emph{Mathematical Programming}, 173\penalty0 (1-2):\penalty0
  465--507, 2019.

\bibitem[Milgrom \& Roberts(1990)Milgrom and Roberts]{Milgrom1990}
Milgrom, P. and Roberts, J.
\newblock Rationalizability, learning, and equilibrium in games with strategic
  complementarities.
\newblock \emph{Econometrica}, 58\penalty0 (6):\penalty0 1255–77, 1990.

\bibitem[Nadav \& Piliouras(2010)Nadav and Piliouras]{nadav2010no}
Nadav, U. and Piliouras, G.
\newblock No regret learning in oligopolies: cournot vs. bertrand.
\newblock In \emph{Proceedings of the 11th ACM Conference on Economics and
  Computation (EC 2010)}, pp.\  300--311, 2010.

\bibitem[Perkins et~al.(2017)Perkins, Mertikopoulos, and
  Leslie]{perkins2017mixed}
Perkins, S., Mertikopoulos, P., and Leslie, D.
\newblock Mixed-strategy learning with continuous action sets.
\newblock \emph{IEEE Transactions on Automatic Control}, 62\penalty0
  (1):\penalty0 379--384, 2017.

\bibitem[Phade \& Anantharam(2019)Phade and Anantharam]{phade2019optimal}
Phade, S. and Anantharam, V.
\newblock Optimal resource allocation over networks via lottery-based
  mechanisms.
\newblock In \emph{International Conference on Game Theory for Networks
  (Gamenets 2019)}, pp.\  51--70. Springer, 2019.

\bibitem[Rosen(1965)]{rosen1965existence}
Rosen, J.~B.
\newblock Existence and uniqueness of equilibrium points for concave n-person
  games.
\newblock \emph{Econometrica}, pp.\  520--534, 1965.

\bibitem[Roughgarden(2016)]{roughgarden2016twenty}
Roughgarden, T.
\newblock \emph{Twenty lectures on algorithmic game theory}.
\newblock Cambridge University Press, 2016.

\bibitem[Serrino et~al.(2019)Serrino, Kleiman-Weiner, Parkes, and
  Tenenbaum]{serrino2019finding}
Serrino, J., Kleiman-Weiner, M., Parkes, D.~C., and Tenenbaum, J.
\newblock Finding friend and foe in multi-agent games.
\newblock In \emph{Proceedings of the 33rd International Conference on Neural
  Information Processing Systems (NeurIPS 2019)}, pp.\  1249--1259, 2019.

\bibitem[Shalev-Shwartz(2007)]{shalev2007online}
Shalev-Shwartz, S.
\newblock \emph{Online learning: Theory, algorithms, and applications}.
\newblock PhD thesis, Hebrew University of Jerusalem, 2007.

\bibitem[Stromberg(2015)]{stromberg2015introduction}
Stromberg, K.~R.
\newblock \emph{An introduction to classical real analysis}, volume 376.
\newblock American Mathematical Society, 2015.

\bibitem[Syrgkanis et~al.(2015)Syrgkanis, Agarwal, Luo, and
  Schapire]{syrgkanis2015fast}
Syrgkanis, V., Agarwal, A., Luo, H., and Schapire, R.~E.
\newblock Fast convergence of regularized learning in games.
\newblock In \emph{Proceedings of the 28th International Conference on Neural
  Information Processing Systems (NeurIPS 2015)}, pp.\  2989--2997, 2015.

\bibitem[Szidarovszky \& Yakowitz(1977)Szidarovszky and
  Yakowitz]{szidarovszky1977new}
Szidarovszky, F. and Yakowitz, S.
\newblock A new proof of the existence and uniqueness of the cournot
  equilibrium.
\newblock \emph{International Economic Review}, pp.\  787--789, 1977.

\bibitem[Yao et~al.(2008)Yao, Adler, and Oren]{yao2008modeling}
Yao, J., Adler, I., and Oren, S.~S.
\newblock Modeling and computing two-settlement oligopolistic equilibrium in a
  congested electricity network.
\newblock \emph{Operations Research}, 56\penalty0 (1):\penalty0 34--47, 2008.

\bibitem[Zhang et~al.(2015)Zhang, Johari, and Rajagopal]{zhang2015competition}
Zhang, B., Johari, R., and Rajagopal, R.
\newblock Competition and coalition formation of renewable power producers.
\newblock \emph{IEEE Transactions on Power Systems}, 30\penalty0 (3):\penalty0
  1624--1632, 2015.

\bibitem[Zhou et~al.(2017)Zhou, Mertikopoulos, Moustakas, Bambos, and
  Glynn]{zhou2017mirror}
Zhou, Z., Mertikopoulos, P., Moustakas, A.~L., Bambos, N., and Glynn, P.
\newblock Mirror descent learning in continuous games.
\newblock In \emph{2017 IEEE 56th Annual Conference on Decision and Control
  (CDC 2017)}, pp.\  5776--5783, 2017.

\bibitem[Zhou et~al.(2018)Zhou, Mertikopoulos, Athey, Bambos, Glynn, and
  Ye]{zhou2018learning}
Zhou, Z., Mertikopoulos, P., Athey, S., Bambos, N., Glynn, P., and Ye, Y.
\newblock Learning in games with lossy feedback.
\newblock In \emph{Proceedings of the 32nd International Conference on Neural
  Information Processing Systems (NeurIPS 2018)}, pp.\  5140--5150, 2018.

\end{thebibliography}
\bibliographystyle{icml2020}

\newpage
\section*{Appendix A}
\subsection*{A1. Review of FKM}
FKM~\cite{flaxman2005online} is a well-known zeroth-order no-regret algorithm under the single function evaluation situation, which is also known as ``gradient descent without a gradient''. 

The pseduocode of FKM is given in Algorithm~\ref{alg:FKM} (reproduced from~\cite{hazan2016introduction}).
\begin{algorithm}[ht]
	\caption{FKM~\cite{hazan2016introduction}}
	\label{alg:FKM}
	\begin{algorithmic}
		\State {\bfseries Input:} decision set $\mathcal{X}$, parameters $\delta$, $\eta$.
		\State Pick $\mathbf{y}_1 = \mathbf{0}$ (or arbitrarily).
		\For{$t=1, 2, ..., T$}
		\State Draw $\mathbf{u}_t \in S_d$ uniformly at random, and set $\mathbf{x}_t = \mathbf{y}_t + \delta \mathbf{u}_t$.
		\State Play $\mathbf{x}_t$ suffer loss $f_t(\mathbf{x}_t)$.
		\State Calculate $\mathbf{g}_t = \frac{n}{\delta} f_t(\mathbf{x}_t) \mathbf{u}_t$.
		\State Update $\mathbf{y}_{t+1} = \prod_{\mathcal{X}_\delta}[\mathbf{y}_t-\eta \mathbf{g}_t]$.
		\EndFor
	\end{algorithmic}
\end{algorithm}


\subsection*{A2. Review of Online Mirror Descent}
Online mirror descent (OMD)~\cite{shalev2007online} is a widely adopted first-order no-regret algorithm, which has been extensively studied under the learning in games setting~\cite{zhou2017mirror,zhou2018learning,mertikopoulos2019learning}. 

The pseudocode of OMD is provided in Algorithm~\ref{alg:OMD} (reproduced from~\cite{hazan2016introduction}).
\begin{algorithm}[ht]
	\caption{Online Mirror Descent with Quadratic Regularization~\cite{hazan2016introduction}}
	\label{alg:OMD}
	\begin{algorithmic}
		\State {\bfseries Input:} decision set $\mathcal{X}$, parameter $\eta >0$, regularization function $R(x)=\frac{1}{2}||\mathbf{x}||_2^2$ which are strongly convex and smooth.
		\State Pick $\mathbf{y}_1 = \mathbf{0}$ (or arbitrarily) and $\mathbf{x}_1 = \arg \min_{\mathbf{x} \in \mathcal{X}} ||\mathbf{y}_1 -\mathbf{x}||_2^2$.
		\For{$t=1, 2, ..., T$}
		\State Play $\mathbf{x}_t$.
		\State Observe the payoff function gradient $\nabla f_t(\mathbf{x}_t)$.
		\State Update $\mathbf{y}_t$ according to the rule:
		$$\text{[Lazy version]} \ \ \ \ \mathbf{y}_{t+1} = \mathbf{y}_{t} - \eta \nabla f_t(\mathbf{x}_t) $$
		$$\text{[Agile version]} \ \ \ \ \mathbf{y}_{t+1} = \mathbf{x}_{t} - \eta \nabla f_t(\mathbf{x}_t) $$
		\State Project to feasible set:
		$$\mathbf{x}_{t+1} = \arg \min_{\mathbf{x} \in \mathcal{X}} ||\mathbf{y}_{t+1} -\mathbf{x}||_2^2$$
		\EndFor
	\end{algorithmic}
\end{algorithm}

\subsection*{A3. Implementation Details}
For the FKM implementation, we have $\mathcal{X} = \mathbb{R}^{+}$ as the feasible action set. We set $\eta_t = \frac{\eta_0}{(0.1t)^{3/4}}$, $\delta = \frac{\delta_0}{t^{1/3}}$, and $\eta_0 = 0.05, \delta_0 = 1$ for all experiments.
For the OMD implementation, we use the agile version update and quadratic regularization. Similarly, the action feasible set is $\mathcal{X} = \mathbb{R}^{+}$. We set $\eta = \frac{1}{2\sqrt{T}}$ as suggested in~\cite{hazan2016introduction} Theorem 5.6 and $T=1000$.
All experiments were run on a MacBook Pro with 16 GB 2400 MHz DDR4 memory, and a 2.2GHz Intel Core i7 CPU.

\section*{Appendix B}
\subsection*{Proof of Proposition 1}
\begin{proposition_noindex}[Restatement of Proposition and Assumptions]
	A Cournot game satisfying the following assumptions:
	\begin{enumerate}
	\item[1)] The market price function $p$ is concave, strictly decreasing, and twice differentiable on $[0,\ymax]$, where $\ymax$ is the first point where $p$ becomes $0$. For $y>\ymax$, $p(y)=0$. In addition, $p(0)>0$. \hfill (A1)
	\item[2)] The individual cost function $C_i(x_i)$ is convex, strictly increasing, and twice differentiable, with $p(0)>C_i'(0)$, for all $i$. \hfill (A2)
   \end{enumerate}
	has exactly one Nash equilibrium.
\end{proposition_noindex}
\begin{proof}
	Define the total production $s=\sum_{i=1}^{n}x_i$. For each player $i$, and each $s \geq 0$, define
	\begin{align*}
	x_i(s) &= \begin{cases} x, & \text{such that  } x \geq 0 \text{ and } p(s) = C_i'(x) - xp'(s).\\
	0, &\text{if no such exists}. \end{cases}
	\end{align*}
	Note that $x_i(s)$ is monotone decreasing in $s$ and $x_i(s)$ is continuous in $s$. It is now shown that there is a unique non-negative number $s^{*}$ such that,
	\begin{equation}
	\sum x_i(s^{*}) = s^{*}.
	\end{equation}
	For $\sum x_i(0) \geq 0$ and by the positivity of $C_i'$, $-p'$, $\sum x_i(\xi) = 0 < \xi$ for any $\xi$ such that $f(\xi) = 0$. As $x(s) = \sum_{i=1}^{N} x_i(s)$ is continuous and strictly decreasing for any $s$ such that $x(s) >0$, there must be exactly one $s^{*}$ for which $x(s^{*}) = s^{*}$. By the definition of $x_i(s)$, each $x_i(s)$ maximizes $\pi_i(x_1, ..., x_N) = p(\sum_{j=1}^{N} x_j)x_i-C_i(x_i)$; therefore, $x(s^{*}) = (x_1(s^{*}), ..., x_N(s^{*}))$ is an equilibrium point of the model, and no other point can be an equilibrium point. The above proof is reproduced from ~\cite{szidarovszky1977new} for reader's easy reference.
\end{proof}
\end{document}